\documentclass[12pt]{llncs}

\usepackage{fullpage}
\usepackage{graphicx}

\usepackage{epsf}
\usepackage{epsfig}
\usepackage{amsfonts}
\usepackage{amssymb}
\usepackage{latexsym}
\usepackage[all]{xy}
\usepackage[boxed,linesnumbered]{algorithm2e}
\usepackage{float}

\newtheorem{fact}{Fact}

\newcommand{\Log}{\mbox{{\sf L}}}

\newcommand{\NL}{\mbox{{\sf NL}}}

\newcommand{\pt}{\mbox{$\sf \oplus$}}

\newcommand{\SL}{\mbox{{\sf SL}}}

\newcommand{\ACo}{\mbox{{\sf AC}$^1$}}

\newcommand{\commentout}[1]{}

\SetKwInOut{Input}{Input}
\SetKwInOut{Output}{Output}
\SetKwInOut{Promise}{Promise}

\begin{document}
\title{Planarity Testing Revisited}
\author{Samir Datta, Gautam Prakriya}
\institute{
  Chennai Mathematical Institute\\
  India\\
  \email{\{sdatta,gautam\}@cmi.ac.in}
}
\maketitle
\begin{abstract}
Planarity Testing is the problem of determining whether a given graph
is planar while planar embedding is the corresponding construction problem.
The bounded space complexity of these problems has been
determined to be exactly Logspace by Allender and Mahajan \cite{AM00} with the aid of Reingold's
result \cite{Rei08}. Unfortunately, the algorithm is quite daunting and 
generalizing it to say, the bounded genus case seems a tall order. 

In this work, we present a simple  planar embedding algorithm running in logspace. 
We hope this algorithm will be more amenable to generalization. The algorithm is based
on the fact that $3$-connected planar graphs have a unique embedding, a  variant
of Tutte's criterion on conflict graphs of cycles and an explicit 
change of cycle basis.% for planar graphs.

We also present a logspace algorithm to find obstacles to planarity, viz. a
Kuratowski minor, if the graph is non-planar. To the best of 
our knowledge this is the first logspace algorithm for this problem.
\end{abstract}

\section{Introduction}
Planarity Testing, the problem of determining whether a given graph
is planar (i.e. the vertices and edges
can be drawn on a plane with no edge intersections except at their 
end-points) is a fundamental problem in algorithmic graph theory. 
Along with the problem of actually obtaining a planar embedding, it is a 
prerequisite for many an algorithm designed to work specifically for
planar graphs.

Our focus is on the bounded space complexity of the planarity embedding 
problem because we know
that many graph theoretic problems like reachability \cite{BTV09}, 
perfect matching \cite{DKR10,DKT10},
and even isomorphism \cite{DLN08,DLNTW09} 
have efficient bounded space algorithms when provided graphs embedded on 
the plane.

Almost a decade ago, building on previous work by Ramachandran and Reif
\cite{RR94}, Allender and Mahajan \cite{AM00,AM04}, proved that Planarity
Testing is contained in $\SL$ and is $\Log$-hard. With Reingold's result
\cite{Rei08} proving $\SL = \Log$, this gave a tight complexity
theoretic classification for the problem. This seemed to be the end of
the story as far as the problem of Planarity Testing vis-a-vis the logspace
world was concerned. 

The only catch
was, the algorithm described in the paper was quite complicated - in fact
a simpler {\SL} algorithm was listed as one of the open questions in 
\cite{AM00}. We would be satisfied with a complicate algorithm
if all we were concerned with was pigeonholing the complexity of
the problem. Planarity testing, however, happens to be a fundamental task in 
Topological Graph theory and a first step towards problems like 
toroidicity testing and bounded genus testing (also listed as open problems in
\cite{AM00}).
Therefore, if we are to make any progress towards a tighter classification of 
these problems, for which no efficient bounded space algorithm is known,
especially, if making a non-blackbox use of a planarity algorithm, 
it is advisable to search for a less daunting algorithm. This work
is the result of this search.

In addition to its simplicity, the interplay between the properties of
$3$-connected planar graphs lends a certain elegance to the algorithm,
at least in our, necessarily, biased eyes.

We also give a logspace algorithm to identify a Kuratowski minor if the 
graph is non-planar. 
%\subsection{Related Work}
\section{Related Work}
Here we survey the related work very briefly - see the paper by Allender and
Mahajan \cite{AM00} for a detailed survey.
The algorithmic aspects of Planarity Testing have been studied since the
inception of Computer Science. It is clear that as far as sequential
computation is concerned, a linear algorithm is optimal. Such algorithms
include the one by Hopcroft and Tarjan \cite{HT74}. The next result
concerns parallel models of computation. Ramachandran and Reif proposed a
very complicated algorithm which worked in logarithmic time and performed 
almost linar work and can be interpreted as placing the problem in the
complexity class \ACo. The final frontier was bounded space computation
and initial sorties had already taken place in the early eighties.
Reif \cite{Reif84} proved that planarity testing for degree $3$ graphs
is in the {\SL}-hierarchy while Ja'Ja' and Simon \cite{JS82a,JS82b} proved that
planarity testing is in the {\NL}-hirerarchy. After Nisan and TaShma's result
\cite{NT95} and the Immerman-Szelpcsenyi Theorem \cite{Immerman88,Szelepcsenyi88} these
bounds become {\SL} and {\NL} respectively. The paper by Allender and Mahajan
completed the campaign when they proved that Planarity Testing was
{\SL}-complete and Reingold's result \cite{Rei08} set the final seal
by proving ${\SL} = {\Log}$.

\section{Outline of Algorithms}
\subsection{Outline of Planar Embedding Algorithm}
We now motivate and informally describe Algorithm~\ref{algo:planarityAlgo}.
It is easy to to see that for a planar graph $G$ and any of its cycles,
for the conflict graph w.r.t. the cycle is bipartite - the bipartitions being
given by the bridges that are placed inside the cycle and those that are
placed outside it, in some planar embedding. 

Conversely,
Tutte \cite{Tut58} has shown that a graph is planar iff the conflict graph
with respect to \emph{every} cycle is bipartite. But since a graph potentially
has exponentially many cycles a direct application of this result does
not yield even a polynomial time algorithm. Thus we ought to seek a small
set of cycles such that restricting our attention to the bipartiteness
of their conflict graphs suffices to yield a planarity test. The set of
fundamental cycles w.r.t. some spanning tree seems to be a natural candidate.
Unfortunately, there are non-planar graphs which have a spanning tree such
that the conflict graph w.r.t. each of the fundamental cycles is bipartite.
e.g. in a $K_5$ with a spanning tree being a star centered at any vertex,
the conflict graph w.r.t. any fundamental cycle has a single bridge.

Suppose, however that we are able to find a deterministic algorithm that 
constructs a valid planar embedding whenever its input is %a $3$-connected
planar graph and either fails when supplied with a non-planar graph or outputs 
a non-planar embedding.
 Now, since, its easy to check (via verification of Euler's
formula) that an \emph{embedding} is planar, we can eliminate any non-planar
graphs at this stage.

Thus it suffices to focus on finding an embedding algorithm that works 
correctly under the promise that the given graph is planar. 
We can without much loss of generality, strengthen the promise and assume 
that the graph is, in addition to being planar, also $3$-connected.
This is because finding the triconnected components is known to be in
logspace (see Lemma~\ref{lemma:dntw} and so is patching together the 
given planar embeddings of triconnected graphs 
(see Lemma~\ref{lemma:biconnToTriconn}).

Concentrating on $3$-connected planar graphs we observe that:
\begin{itemize}
\item The graph has a unique planar embedding. 
\item The conflict graphs w.r.t. any cycle is connected 
(see Section~\ref{sec:conflictGraph}),
which enables us to bipartition the bridges in a unique way so that bridges
in a partition lie on one side of the cycle in question.
\item Ideally we would like to pick a small set of cycles and determine
which edges lie inside them and somehow piece together the combinatorial
embedding from this information. A natural choice for such a set is 
the set of fundamental cycles w.r.t. an arbitrary spanning tree. But the
following two problems crop up:
	\begin{itemize}
	\item Though we have a bipartition of bridges for each fundamental cycle
		, it is not clear which partition is mapped inside and 
		which outside.
	\item How do we actually piece together a combinatorial embedding
		once we know this information?
	\end{itemize}
\item The first problem can be solved if we knew at least one face of the 
unique embedding, because then we can just think of the face as the external
face and stipulate that every edge (not part of the face) is contained inside
it. Thus for every fundamental cycle, the bridge containing all the edges
of the external face would lie outside the cycle - this fixes which bridges
lie inside and which outside. But Corollary~\ref{cor:fundFace} tells us
that for the graphs in question there is always a fundamental cycle which
is a face. This, combined with Fact~\ref{fact:faceIndNonsep} ensures that
we can find such a face in logspace.
\item We solve the second problem in Section~\ref{section:embed}, where
we first show that finding a solution is equivalent to a change of the cycle
basis from fundamental cycles to faces. Then we show an explicit way to 
perform this inversion in logspace.
\item Finally, we know that a $3$-connected graph is planar iff every edge 
lies on exactly two induced non-separating cycles (see 
Fact~\ref{fact:faceIndNonsep}). Thus, 
we can use this criterion to check for planarity of $G$.
\end{itemize}

\subsection{Outline of Algorithm to find Kuratowski Minors}
In Section~\ref{section:kuratowski} we describe an algorithm 
(Algorithm~\ref{algo:kuratowskiAlgo}) to find
a Kuratowski (i.e. a $K_5$ or $K_{3,3}$) minor in a non-planar graph.
The algorithm indentifies a cycle with a non-bipartite
conflict graph. It then finds
an induced odd cycle in this conflict graph. Finally, it contracts some
of the bridges and vertices of attachment of these bridges to yield
a Kuratowski minor.

\section{Definitions and Preliminaries}
We assume familiarity with basic complexity theory in general and
bounded space classes in particular, as described
in any standard text e.g \cite{AroraB09}. We will also assume
familiarity with graph theory as described in a text like 
\cite{Diestel05,west00}. Below, we explicate all non-standard material
we will have occasion to use.

\begin{definition} The bridges of a cycle $C$ consist of:
\begin{itemize}
\item For every connected component $X$ of $G \setminus C$, the induced graph
$G[X \cup A_X]$ where $A_X \subseteq C$ are the vertices of $C$ adjacent to
some vertex of $X$ (the so called points of attachment).
\item The chords of $C$ - here the endpoints of $C$ are its points of attachment
\end{itemize}
\end{definition}
\begin{definition}
Two bridges $B_1,B_2$ of a cycle $C$ conflict iff either of the following
conditions hold:
\begin{itemize}
\item $a_i,a'_i$ are two points of attachment of $B_i$ w.r.t. $C$ 
for $i \in \{1,2\}$ such that they occur in the order $a_1,a_2,a'_1,a'_2$ 
along the cycle $C$.
\item $B_1,B_2$ have three common points of attachment w.r.t. cycle $C$.
\end{itemize}
We can extend the definition of conflict to sets of bridges 
$\mathcal{B}_1,\mathcal{B}_2$ by the existence of two sets of 
conflicting vertices of attachment in the above sense, belonging to
(not necessarily the same bridge of) $\mathcal{B}_1,\mathcal{B}_2$ respectively.
The conflict graph $H_C(G)$ of a graph w.r.t. a cycle $C$ is formed by taking the bridges of
$C$ as vertices and joining two vertices by an edge iff they conflict.
\end{definition}
\begin{definition} Given a spanning tree $T$ of a biconnected graph $G$,
and an edge $e \in E(G) \setminus E(T)$, the graph $G \cup e$ contains a 
unique cycle $C(e)$ called the fundamental cycle of $e$. We say that a face
of an embedded planar graph is fundamental if it is a fundamental cycle
of some non-tree edge (with respect to some fixed spanning tree).
\end{definition}
The following is an easy consequence of Reingold's result \cite{Rei08}
since we can count the number of components in $(G \cup e) \setminus e'$.
We single it out of a set of similar elementary graph computations
summarized in Section~3.2 of \cite{AM00} which can be done in logspace
because it is of special significance for us.
\begin{fact}\label{fact:fundCycle}
The list of edges in each fundamental cycle of $G$ w.r.t. a spanning tree $T$
can be obtained by a logspace transducer.
\end{fact}

\begin{fact}\label{fact:faceIndNonsep} (Proposition~4.2.10 and Theorem~4.5.2 \cite{Diestel05})
 The faces of a $3$-connected planar graph $G$ are exactly the
induced non-separating cycles of $G$. 
Further, a $3$-connected graph
is planar iff every edge lies on exactly two induced, non-separating cycles.
\end{fact}

\begin{proposition}\label{prop:dualToPrimalEmbed}
Given a the cyclic order of vertices in every face of a biconnected
embedded graph, it is possible to construct in logspace, a combinatorial 
embedding of the graph.
\end{proposition}
\begin{proof}
Call an ordered triplet,$T$ of vertices $(u,v,w)$ an angle of an embedded
graph if $(u,v),(v,w)$ are consecutive edges on some face of the embedding.
Call two triplets $T_i = (u_i,v_i,w_i)$ for $i \in \{1,2\}$ adjacent if
$v_1 = v_2$ and either $w_1 = u_2$ or $u_1 = w_2$. Then it is
clear that the undirected graph with angles as vertices and angle adjacencies
as edges, forms a set of disjoint\footnote{since any angle has a unique
middle vertex} cycles\footnote{since every angle is adjacent to 
exactly two angles on a vertex}. The ordering on vertices induced by the
angles is exactly the combinatorial embedding. \qed
\end{proof}

\section{Planarity Testing}
\subsection{Reduction to the Triconnected case}
\begin{lemma}\label{lemma:dntw} (Lemma 3.3 \cite{DNTW09}) 
The triconnected components of a graph can be obtained in logpace
\end{lemma}
It is a well known fact that if a graph is non-planar then one of it's $3$-connected components 
is non-planar. \cite{west00}

\begin{lemma}\label{lemma:biconnToTriconn}
Given a combinatorial planar embedding of the triconnected components of a graph, it is possible to 
obtain the biconnected planar embedding of the graph in logspace.
\end{lemma}
\begin{proof}
Consider the tri-connected component , separating pair tree.
This tree can be constructed in logspace \cite{DNTW09}.
We will use the tree as an index to construct the embedding of the graph, given 
the embedding of it's $3$ components.  
Given a vertex in a tri-connected component,we can obtain the clockwise order 
of the edges around the vertex using the given combinatorial embedding.

Given a  node in the tree corresponding to separating pair $u,v$ we arrange the tri-connected components 
around $u$ in the order they appear in the tree and around $v$ in the opposite order. This gives us an ordering of 
edges around each vertex in the graph.  
We repeat this procedure for each of the nodes corresponding to a separating-pair in the tree.
This gives us an embedding of the graph. To move from biconnected graphs to general graphs, see \cite{AM00}.\qed
\end{proof}

\subsection{The Conflict Graph}\label{sec:conflictGraph}
\begin{lemma}\label{lemma:conflictGraph}
 Given a $3$-connected planar graph $G$ and an arbitrary cycle
$C$ in the graph the conflict graph $H_C(G)$ is bipartite and connected.
\end{lemma}
\begin{proof}
It is easy to see that for an embedded planar graph $G$ and any
cycle $C$ contained in it, the conflict graph, $H_C(G)$ is bipartite (with
the bipartition being given by whether a bridge is mapped ``inside'' or 
``outside'' the cycle). 
We now show that $H_C(G)$ is connected. 

For every pair of points $x,y \in C$ there are two paths from x to y along C. We will call these paths $P_{x,y}$ and $P'_{x,y}$.
Let ${\mathcal{B}}$ be a set of bridges that correspond to one of the connected components of $H_C(G)$.

Pick a pair $u,v$ s.t
 \begin{itemize}
\item$u$ and $v$ are attachment points of bridges in $\mathcal{B}$.
\item$u$ and $v$ are not adjacent.
\item The attachment points of all bridges in  $\mathcal{B}$ lie in either $P_{u,v}$ or $P'_{u,v}$. (Assume WLOG that they lie in $P_{u,v}$).
 \end{itemize}
We now show that one can always find such a pair $u,v$.
Pick a bridge $B'  \notin {\mathcal{B}}$ . The attachment points of $B'$ divide $C$ into
a number of segments .(atleast $3$ since every bridge has  $3$ or more points of attachment.)
Since no bridge in ${\mathcal{B}}$ conflicts with $B'$ and the vertices corresponding to ${\mathcal{B}}$
in $H_C(G)$ form a connected component, it follows easily that all points of attachments of ${\mathcal{B}}$ must lie in one of the
segments. This implies that there is a point in $C$(a point of attachment of $B'$) that is not a point of attachment of any bridge in ${\mathcal{B}}$. 
Therefore we can find a pair $u,v$ with the properties listed above.

Now, for every point $x \notin \{u,v\}$ on $P_{u,v}$, $\exists$ a bridge $B \in \mathcal{B}$ with 
points of attachment $b_{1}, b_{2}$ s.t $b_{1}$ precedes $x$ and $b_{2}$ succeeds $x$ in $P_{u,v}$. 
If not then $\forall B \in \mathcal{B}$ either all points of attachment of $B$ lie between $u$ and $x$ 
or all of them lie between $x$ and $v$.
This implies that in $H_C(G)$ vertices corresponding to  bridges in $\mathcal{B}$  $u$ as a point of attachment are 
not connected to vertices corresponding to bridges with $v$ as a point of attachment.
This is a contradiction since ${\mathcal{B}}$ corresponds to a connected component in $H_C(G)$.

Now, if $\{u,v\}$ is not a separating pair then $\exists$ a bridge $B \notin \mathcal{B}$ with points of attachment in both
 $P_{u,v} \setminus \{u,v\}$ and $P_{u,v} \setminus \{u,v\}$.
From the above it follows that $B'$ conflicts with a bridge in  $\mathcal{B}$ which is not possible.
Therefore, $\{u,v\}$ is a separating pair.\qed
\end{proof}
%This implies that for $3$-connected graphs the conflict graph  $H_C(G)$ w.r.t every cycle $C$ is connected.

\subsection{Inside and Outside a Fundamental Cycle}
Next we fix which bridges w.r.t. each of the fundamental cycles are mapped
inside and which are mapped outside the concerned cycle. Basically we
find one \emph{face} of the grap and call it the external face. Thus
every bridge containing this face is mapped outside any other cycle. This
completely fixes the $2$-coloring.
\begin{proposition}\label{prop:fundFace}
 Any biconnected embedded planar graph has a fundamental \emph{face} (i.e.
a fundamental cycle which is also a face)
w.r.t. each of its spanning trees.
\end{proposition}
\begin{proof}
It is well known that the set of dual edges of the non-tree edges
of a biconnected planar graph forms a tree (after reducing all multi-edges
to a single edge) - the so called dual tree. This follows from observing that
cycles in the primal correspond to cuts in the dual (Proposition~4.6.1 \cite{Diestel05}) and therefore the set of edges dual to the non-tree edges are connected
in the dual graph, the proof being completed by applying Euler's relation to
verify that the number of such edges is exactly one less than the number of
faces.  It is easy to see that the
face of the original graph corresponding to a leaf of the dual tree is
a fundamental face. \qed
\end{proof}
A direct consequence of Fact~\ref{fact:fundCycle},
Proposition~\ref{prop:fundFace}, 
Fact~\ref{fact:faceIndNonsep} is the following (since finding whether
a cycle is induced and non-separating is a logspace predicate given
the list of edges in the cycle on the input tape).
\begin{corollary}\label{cor:fundFace} Every $3$-connected planar graph has at least one fundamental
face and this can be found by a Logspace transducer.
\end{corollary}
Knowing which edges of the graph map to the region enclosed by each of the
fundamental cycles, we proceed
to construct a purported embedding of the given graph. If the graph is indeed
planar we will obtain a valid planar embedding, else we will not be able to do
so, giving an effective planarity test (which produces an embedding as a 
side-effect).
\subsection{Obtaining a Planar Embedding}\label{section:embed}
Here we exhibit a way to identify the edges in each face of the given
$3$-connected graph. Fact~\ref{fact:faceIndNonsep} then gives a way of 
effectively checking if the graph is plnar.

At the heart of the proof is a change of basis in the cycle space over 
$\mathbb{Z}_2$ Somewhat msurprisingly %iraculously
the solution to the equations obtained in the previous section tell us
which faces sum up (over $\mathbb{Z}_2$) to give a particular fundamental cycle.
Notice that we have only an ``implicit'' representation of the faces (see
below) and while we are seeking for an explicit representation in terms 
of which edges constitute a face. Since fundamental cycles and internal 
faces both form a basis of the cycle space over $\mathbb{Z}_2$, we just need 
to invert the matrix that expresses the fundamental cycles as a linear 
combination of faces.  Though seemingly this would place the problem in 
\pt\Log\ we give an explicit way to perform this inversion
which yields a \Log\ upper bound.

Returning to the ``implicit'' representation of faces, there is a natural
bijection between non-tree edges and faces viz. one that maps a non-tree edge
$e$ to that face $f(e)$ adjacent to $e$ which lies \emph{inside} the fundamental
cycle $C(e)$ (w.r.t. the external face $C(e_0)$). With faces labeled in this way,
the solution of the preceding section tells us which faces are contained 
within the fundamental cycle $C(e)$. 

Let us start by fixing some notation. For distinct non-tree edges $e_1,e_2$,
define $e_1 \prec e_2$ iff %the fundamental cycle $C(e_2)$ encloses $e_1$.
in a $2$-coloring of the conflict graph $H_{C(e_2)}(G)$, the colors of
the vertices corresponding to bridges containing $e_0,e_1$ get different
colors. Intuitively, this means that (if $G$ is planar) $C(e_2)$ separates 
$e_1$ from $e_0$ in the unique planar embedding of $G$.

%%Given a solution,
%%$\{x_{e',e} | e \in E(G) \setminus E(T) \wedge e' \in E(G) \setminus C(e)\}$ of the equations in the previous section, for each non-tree edge 
For each $e \in E(G) \setminus \{e_0\}$, let $P(e)$ denote:
\[
\{ e' \in E(G) \setminus (E(T) \cup \{e\}) | e' \prec e \wedge \nexists e'': e'\prec e''\prec e \},
\]
Further, let $F(e)$ denote 
\[
\bigoplus_{e' \in P(e)\cup\{e\}}{ C(e')} 
\]
Notice that for sets $A,B$ the set $A \oplus B$ denotes the symmetric 
difference of the two sets and the notation above refers to an iteration of 
this operation.

If the graph is $3$-connected planar then in its unique planar embedding
with $C(e_0)$ as the external face, $P(e)$ consists  of
the set of non-tree edges 
that are enclosed by $e$ but not by any $e''$ which is also enclosed by $e$.
Intuitively, it clear that these are exactly the non-tree edges
occurring on the face $f(e)$ (see description in the preceding paragraph).
We will make this intuition precise and in fact, show the following:
\begin{lemma} For each non-tree edge $e \neq e_0$, the face $f(e)$ consists
exactly of the edges in the set $F(e)$.
\end{lemma}
\begin{proof} Given a fundamental cycle $C$ let $\phi(C)$ represent the number of faces 
that lie inside $C$ and $\psi(C)$ represent the number of fundamental cycles 
that lie inside $C$ including $C$. Notice that for fundamental cycles
which are also faces, $\phi(C) = \psi(C) = 1$. In fact for every fundamental
cycle $C$, as an easy consequence of Euler's formula it follows that
$\phi(C) = \psi(C)$.
 
Let $e \in E(G) \setminus E(T)$. We will first show that $F(e)$ is a face.
Since each $e' \in P(e)$ lies inside $C(e)$ and for every non-tree 
$e'' \not\in  P(e)$, such that $e'' \prec e$, there is exactly one $e' \in P(e)$
such that $e'' \prec e'$. There is one such $e'$ because we know that
every non-tree edge enclosed by $C(e)$ is either enclosed by some other
fundamental cycle $C(e_1)$ or otherwise is in $P(e)$. Thus by induction
on the ``enclosure depth'' of an edge we get an $e' \in P(e)$ which encloses
it. The uniqueness follows from observing that if both $e \prec e' \wedge e \prec e'_1$ then either $e'_1 \prec e'$ or $e' \prec e'_1$.
 - this is due to planarity - therefore both $e',e'_1$ cannot be in $P(e)$.
Thus we get,
\[
\psi(C(e)) = 1 + \sum_{e' \in P(e)}{\psi(C(e'))}
\]
therefore:
\[
\phi(C(e)) = 1 + \sum_{e' \in P(e)}{\phi(C(e'))}
\]
Now since every fundamental cycle can be written as a sum of internal faces
of $G$ and $F(e)$ is a sum of a set of fundamental cycles (where all 
computation is over $\mathbb{Z}_2$), it is also a sum of some internal faces.
Since only faces contained within $C(e)$ figure in this sum, $F(e)$ must
be a disjoint sum of some cycles enclosed by $C(e)$. We can be more
specific, the sum:
\[
\bigoplus_{e' \in P(e)}{C(e')}
\]
includes exactly the faces inside $C(e)$ which are also faces inside some 
$C(e')$ for $e' \in P(e)$. Thus $F(e)$ includes exactly
the faces contained inside $C(e)$ but not in any $C(e')$ for 
$e' \in P(e)$. Now, it is easy to see from the expression for
$\phi(C(e))$ that there is exactly one such face, so it follows from the
linear independence of the faces that this must be $f(e)$. \qed
\end{proof}

Thus, we can complete the proof of the following:
\begin{theorem}\label{thm:main}
Given a graph $G$, constructing a planar embedding for $G$ if it is planar
and otherwise rejecting it, can be done in logspace.
\end{theorem}
\begin{proof}
Given a graph we obtain its $3$-connected components using Lemma~\ref{lemma:dntw}.
If each triconnected compoenent is planar, we will 
successfully obtain a planar embedding for each component and then
obtain the proof with the aid of Lemma~\ref{lemma:biconnToTriconn}.
If some triconnected component is not planar, we will either obtain an
$F(e)$ which is not induced or is separating or we will obtain an edge
lying on at least three $F(e)$'s and therefore reject.
\end{proof}

%\IncMargin{1em}
\begin{algorithm}[h!]
\Input{Graph $G$}
\Promise{$G$ is $3$-connected}
\Output{A planar embedding of $G$ if planar and $\emptyset$ otherwise}
\BlankLine
Compute a spanning tree $T$ in $G$\;
\lIf{there is no fundamental face w.r.t. $T$}{\Return{$\emptyset$}}\;
Let $C_0 = C(e_0)$ be a fundamental face\;
\ForEach{fundamental cycle $C(e) \neq C_0$}{
	Construct conflict graph $H_{C(e)}(G)$\;
	\lIf{$H_{C(e)}(G)$ is not bipartite}{\Return{$\emptyset$}}\;
	$2$-color $H_{C(e)}(G)$\;
	\ForEach{non-tree edge $e' \neq e$} {
		\tcp{Let $B_{e',e}$ be the bridge of $C(e)$ containing $e'$}
		\lIf{$B_{e',e}$ gets color different from $B_{e_0,e}$ in $H_{C(e)}(G)$}{let $e' \prec e$}\;
	}
}
\ForEach{fundamental cycle $C(e) \neq C_0$}{
	Let
	$ P(e) = \{ e' \in E(G) \setminus (E(T) \cup \{e\}) | e' \prec e \wedge \nexists e'': e'\prec e''\prec e \}
	$\;
	Let 
	$ F(e) = \bigoplus_{e' \in P(e)\cup\{e\}}{ C(e')} $\;
}
Let $F(e_0) = C_0$\;
\ForEach{edge $e \in E(G)$}{
	\lIf{$e$ doesn't lie in exactly two $F(e')$'s}{\Return{$\emptyset$}}\;
}
\tcp{Follow the proof of Proposition~\ref{prop:dualToPrimalEmbed} to construct an embedding of $G$}
\Return{planar embedding of $G$}\;
\caption{Planarity Testing and Embedding Algorithm for $3$-connected graphs}
\label{algo:planarityAlgo}
\end{algorithm}%\DecMargin{1em}

\begin{algorithm}[h!]

\Input{Graph $G$}
\Promise{$G$ is non-planar}
\Output{A Kuratowski Minor}
\BlankLine
Find a non-planar subgraph $G'$ and a cycle $C \subseteq G'$ s.t. $H_C(G')$ is non-planar \;
\Begin{
	Fix an ordering of edges of $G$ \;
	Find smallest edge $e_i = (x,y)$ such that graph $G'$ on $\{e_1,\ldots,e_i\}$ is non-planar\;
	Find planar embedding of $G' \setminus e_i$\;
	Let $F_x,F_y$ be two faces containing $x,y$ respectively\;
	Find a simple path between $F_x,F_y$ in $dual(G')$ which avoids all other faces containing $x,y$\;
	Let $C$ be the symmetric difference of the faces along this path\;
	\tcp{$H_C(G')$ is non-bipartite}
}
Find an induced odd cycle $\bar{C}$ in $H_C(G')$\;
\Begin{
	Find a spanning tree of $H_C(G')$ and $2$-color it\;
	Pick a non-tree edge $e$ with both endpoints of same color, so that 
	all the chords of fundamental cycle $C(e)$ have different colors\;
	Walk on the tree between endpoints of $e$ and ``short-circuit''
	as many tree-paths, by chords, as possible\;
}
Find a cycle with $3$ mutually conflicting bridges\;
Contract the bridges to single vertices (if they originally contained a vertex)
or edges\;
Find a Kuratowski minor of this graph by brute force\;
Expand back to a Kuratowski minor of the original graph\;
\caption{Algorithm for finding Kuratowski Minors}
\label{algo:kuratowskiAlgo}
\end{algorithm}%\DecMargin{1em}

\section{Finding Kuratowski Subgraphs} \label{section:kuratowski}
We describe an an algorithm to obtain a Kuratowski subgraph given a cycle with 
a non-bipartite conflict graph.  
As a consequence, to obtain a Kuratowski subgraph, it is sufficient to find 
such a cycle in a subgraph $G'$ of $G$. 
We do this in a series of lemmas:

\begin{lemma}
Given a non-planar graph $G$, we can, in logspace, 
find a non-planar subgraph $G'$ and a cycle $C \subseteq G'$ s.t. $H_C(G')$ is 
non-bipartite.
\end{lemma}
\begin{proof}
Given access to a routine for planarity checking, and a non-planar graph $G$,
we can find a minimal non-planar subgraph $G'$ of $g$ in the following sense.
Order the edges of $G$ arbitrarily as $e_1,\ldots,e_m$. Now consider the 
smallest $i$ such that the graph $G'$ formed by the union of the first $i$ 
edges $e_1,\ldots,e_i$ is non-planar. Notice that this means that 
$G' \setminus e_i$ is necessarily planar.

Now, we show how one can find a cycle in this graph $G'$ such that the conflict
graph $H_C(G')$ is non-bipartite. Construct a planar embedding of 
$G' \setminus \{e\}$. The endpoints, say $x,y$, 
of $e_i$ must lie on different faces of this embedding because $G'$ is non-planar.

Find a path in the dual graph of the embedding between any two faces $F_x,F_y$ incident respectively on $x,y$.
(this path must avoid any other faces incident on $x$ or $y$.)
The symmetric difference of the faces in this path is a cycle $C$ one of whose bridges is $e_i$.

We claim that $H_C(G')$ is non-bipartite. If it were bipartite then it would be possible to
give an orientation to each of it's bridges such that conflicting bridges got opposing orientations.
But since the bridges of $C$ are all planar, this would imply that $G'$ is planar .
Therefore, $H_C(G')$ is non-bipartite.  
\end{proof}

%remove the primal edges corresponding to the dual edges in this path.
%Now $x,y$ must share a common face.
% Let $C$ be the cycle bounding this
 %face if it is not the external face else let $C$ be the cycle bounding
%the external face in the component that $x$ lies. It is easy to see that
%in $G'$ (and therefore in $G$), $C$ must have a non-bipartite conflict graph.
%(Proof?)

\begin{lemma}
Given a non-planar graph $G'$ and a cycle $C$ witnessing this via the
non-bipartiteness of $H_C(G')$, we can, in logspace, 
find an induced odd cycle $\bar{C}$ in $H_C(G')$.
\end{lemma}
\begin{proof}
Since the conflict graph $H_C(G')$ is non-bipartite, it contains an odd cycle.
We aim to find an induced odd cycle in this graph. For this consider a spanning
tree of $H_C(G')$. $2$-color the spanning tree. Notice that there must exist
a non-tree edge between two vertices with the same color (else the graph
would have been bipartite). Find a non-tree edge $e$ such that all the chords
of the fundamental cycle $C(e)$ get opposite colors - a simple induction on the
number of chords that a fundamental cycle has, shows that such an edge must
exist and locating it in logspace is easy. 

We will construct a chordless cycle from $C(e)$ by replacing some tree paths by chords of $C(e)$. To do this,
let $e = (u_1,u_k)$ and let the vertices of the tree path from $u_1$ to $u_k$
be $u_2,\ldots,u_{k-1}$ in order. Call the point $u_i$ of a chord $(u_i,u_j)$
for $i < j$ as the origin of the chord and $u_j$ the end of the chord.
Now start walking from $u_1$ to $u_k$ along the tree path, outputting tree
edges till the origin of a chord is encountered. Output the chord and move 
on to the end of the chord and repeat. It is easy to see that the edges output by this
procedure along with the non-tree edge $(u_1,u_k)$ form an induced cycle
because either the origin or endpoint of any other chord of $C(e)$
 is not on this cycle. Also, because the endpoints of all edges on this cycle are oppositely
colored except for $(u_1,u_k)$, this is an odd cycle.
\end{proof}

At this point we have a cycle $C$ in $G$ with a non-bipartite conflict graph
$H_C(G')$ and an induced odd cycle $\bar{C}$ in $H_C(G')$ witnessing this. 
Now we prove the following,
\begin{lemma} 
Given an induced odd cycle $\bar{C}$ in $H_C(G')$, we can, in logspace, find a Kuratowski subgraph.
\end{lemma}
\begin{proof}
First, suppose that two conflicting bridges $B_{1}$ and $B_{2}$ in the odd cycle share $3$ points of attachment $a_{1},a_{2}$ and $a_{3}$. 
If either of the bridges( $B_{1}$ say.) has another  point of attachment $a_{4}$ then, clearly the points of attachment $a_{1},a_{3}$ of $B_{2}$
alternate with the points of attachment $a_{2},a_{4}$ of $B_{1}$.
Therefore, it is sufficient to deal with the case when the case where both $B_{1}$ and $B_{2}$ have just $3$ points of attachment.
In this case it is not difficult to see that 
any bridge that conflicts with $B_{1}$ also conflicts with $B_{2}$ 
and vice-versa. Hence $\bar{C}$ must be a $3$-cycle. We deal with this case in lemma~\ref{lemma:TriangleConflict}.

Next, consider the case where $\bar{C}$ is an odd cycle of size greater than 
$3$ and the conflict of a bridge $B$ in the odd cycle 
with any other bridge in the odd-cycle is witnessed by $2$ 
 points of attachment of $B$ i.e we exclude the possibility
that $2$ bridges share $3$ points of attachment.

It is easy to see that $\exists$ $2$ points of attachment for every bridge that witness both it's conflicts.
Thus by contracting all bridges (excluding points of attachment) to single points and removing edges of attachment, one can
reduce all bridges to paths while maintaining all conflicts.
This can be done in logspace since for each bridge one only needs to remember the witnesses of conflict of the bridge with its neighbours (in $\bar{C}$).(atmost 8 points.)

If we label the vertices of $\bar{C}$ as $v_{B_0},v_{B_1},\ldots,v_{B_{2k}}$ $(k > 1)$
for some positive integer $k$, and the points of attachment of bridge $B_i$
as $u_i,v_i$ then the points occur along $C$ in the order:
$u_0, v_{2k}, u_1, v_0, u_2, v_1, \ldots, u_{2k-1}, v_{2k-2}, u_{2k}, v_{2k-1}$.

\begin{figure}\label{cycle7}
\centering
%\begin{tabular}{cc}
\scalebox{0.5}{\includegraphics[width = \textwidth]{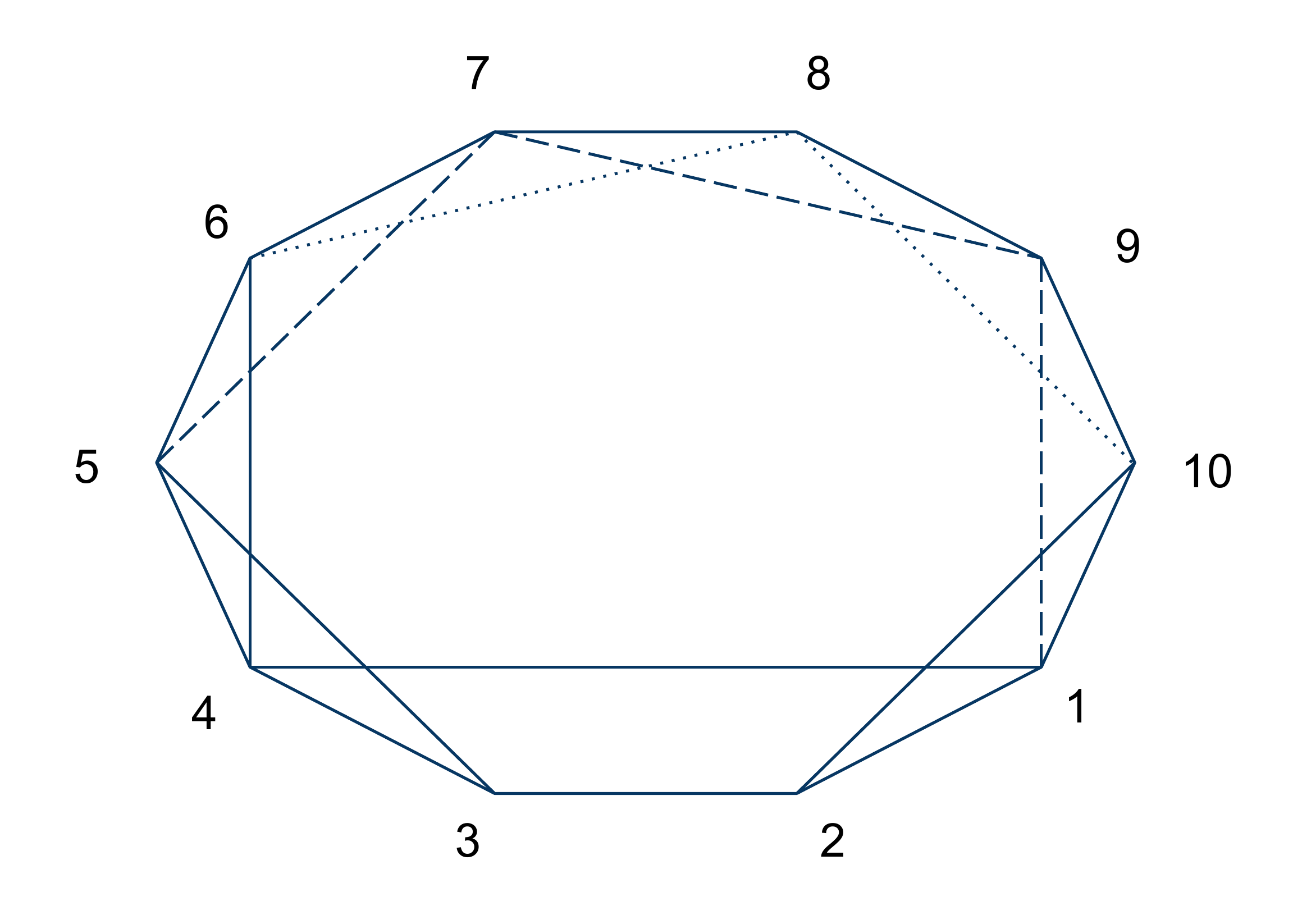}}
%\epsfig{file=gridsign.eps,width=0.5\linewidth,clip=} &
%\epsfig{file=gridweight.eps,width=0.5\linewidth,clip=} \\
%\end{tabular}
\caption{In the figure, we have a cycle with the $9$-cycle as it's conflict graph .
On treating $\{6,8,10\}$ as a single point and leaving out the edges $(6,7),(7,8),(8,9)$ and $(9,10)$, we obtain a subdivision of $K_5$}

\end{figure}
Now, consider the cycle $u_{2k},u_{0},v_{2k},v_{0} ,U'$ 
(Where $U'=\{\{v_1, u_3\}\cup B_3\cup \{v_3\}\cup\ldots \cup \{u_{2k-1}\}$)
along with the bridges  $B_{2k}$ , $B_{0}$ the paths connecting $u_{0}$ and $U'$,
$v_{2k}$ and $U'$ and the path $\{v_{0},u_{2}\} \cup B_{2} \ldots \{v_{2k-2},u_{2k}\}$
Clearly, this is a $K_5$ minor.(see Fig.1 for an example.)
\end{proof}

Finally we are left with the following case: % of three mutually conflicting bridges:
\begin{lemma}\label{lemma:TriangleConflict}
Given a cycle in $G'$ with three mutually conflicting bridges, it is possible
to extract a Kuratowski minor of $G'$ in logspace.
\end{lemma}
\begin{proof}
%We first deal with the case where the odd cycle $\bar{C}$ is a triangle.
Clearly, it is sufficient to consider $4$ points of attachment for each bridge since 
$4$ points of attachment suffice to witness conflict with two other bridges.
We can contract the bridges (excluding the points of attachment) to single points and reduce the
problem of finding a Kuratowski subgraph in the above graph to that of finding a Kuratowski minor in a
non-planar graph with atmost $15$ vertices and $24$ edges. 
Since this is a constant sized graph we can find a Kuratowski minor by
brute-force.
%One can do this by using any of the standard algorithms for the problem.
\end{proof}
%We now show how one can reduce all other cases to the above case.

\section*{Acknowledgements}
The first author would like to acknowledge discussions with Eric Allender, 
Abhishek Bhrushundi, Ashish Kabra, Raghav Kulkarni, 
Nutan Limaye, Meena Mahajan, and Prajakta Nimbhorkar on this topic.
\bibliographystyle{alpha}	
\bibliography{planaritytest}
\commentout{
\appendix
\input{appendixProofs}
}
\end{document}